\documentclass[showpacs,aps,twocolumn,preprintnumbers,superscriptaddress,longbibliography]{revtex4-1}

\usepackage{amsmath,amssymb,amsthm}
\usepackage[dvipdfmx]{graphicx}

\newtheorem{theorem}{Theorem}
\newtheorem{assumption}{Assumption}
\newtheorem{lemma}{Lemma}

\def\bra#1{\langle #1|}
\def\ket#1{|#1\rangle}
\def\braket#1{\langle #1\rangle}

\newcommand{\mrm}[1]{\mathrm{#1}}
\newcommand{\rmS}{{\mathrm{S}}}

\newcommand{\rmB}{{\mathrm{B}}}
\newcommand{\rmI}{{\mathrm{I}}}

\newcommand{\prob}{\mathrm{Prob}}

\newcommand{\tr}{\mathrm{Tr}}

\newcommand{\abs}[1]{\left|#1\right|}

\newcommand{\onorm}[1]{\left\|#1\right\|}

\newcommand{\trdist}[2]{\mathcal{D}(#1,#2)}
\newcommand{\ldist}[2]{\mathcal{D}_{\rmS}(#1,#2)}

\newcommand{\hs}{H_\rmS}
\newcommand{\hi}{H_\rmI}
\newcommand{\hb}{H_\rmB}

\newcommand{\rmc}{\rho^{\mathrm{mc}}}

\newcommand{\rs}{\rho_\mathrm{S}}

\newcommand{\rb}{\rho_{\mathrm{B}}}
\newcommand{\rbcan}{\rho^{\mathrm{can}}_{\mathrm{B}}}
\newcommand{\rbmc}{\rho^{\mathrm{mc}}_{\mathrm{B}}}

\newcommand{\rde}{\omega}
\newcommand{\rcande}{\omega^{\mathrm{can}}}
\newcommand{\dmccan}{\delta_{\mrm{eq}}}

\newcommand{\nde}{D_{\mathrm{E}}}
\newcommand{\ndg}{D_{\mathrm{G}}}
\newcommand{\deff}{D_{\mathrm{eff}}}
\newcommand{\dbmc}{D^{\mathrm{mc}}_{\mathrm{B}}}

\newcommand{\ssy}{S_{\rmS}}
\newcommand{\epde}{\braket{\sigma_{\infty}}}
\newcommand{\epcande}{\braket{\sigma_{\infty}^{\mathrm{can}}}}
\newcommand{\dep}{\delta_{\mrm{EP}}}
\newcommand{\dei}{\delta E_{\mathrm{I}}}

\begin{document}


\title{Saturation of entropy production in quantum many-body systems}


\author{Kazuya Kaneko}
\affiliation{
	Department of Applied Physics, The University of Tokyo,
	7-3-1 Hongo, Bunkyo-ku, Tokyo 113-8656, Japan
}


\author{Eiki Iyoda}
\affiliation{
	Department of Applied Physics, The University of Tokyo,
	7-3-1 Hongo, Bunkyo-ku, Tokyo 113-8656, Japan
}
\author{Takahiro Sagawa}
\affiliation{
	Department of Applied Physics, The University of Tokyo,
	7-3-1 Hongo, Bunkyo-ku, Tokyo 113-8656, Japan
}

\date{\today}

\begin{abstract}
Bridging the second law of thermodynamics and microscopic reversible dynamics has been a longstanding problem in statistical physics.
Here, we address this problem on the basis of quantum many-body physics,
and discuss how the entropy production saturates in isolated quantum systems under unitary dynamics.
First, we rigorously prove that
the entropy production does indeed saturate in the long time regime,
even when the total system is in a pure state.
Second, we discuss the non-negativity of the entropy production at saturation,
implying the second law of thermodynamics.
This is based on the eigenstate thermalization hypothesis (ETH),
which states that even a single energy eigenstate is thermal.
We also numerically demonstrate that 
the entropy production saturates at a non-negative value
even when the initial state of a heat bath is a single energy eigenstate.
Our results reveal fundamental properties of the entropy production in isolated quantum systems at late times.
\end{abstract}


\maketitle


\section{\label{sec:Introduction}Introduction}
In the mid-nineteenth century, Rudolf Clausius introduced entropy as a measure of thermodynamic irreversibility,
by arguing that, while the energy of the universe is constant, the entropy of the universe tends to a maximum~\cite{Clausius1867}.
Ever since, understanding inevitable irreversibility in the macroscopic world is one of the most fundamental problems in statistical physics.
Especially, in light of the fact that isolated quantum systems have reversible dynamics,
it has not been clear whether such a behavior of the entropy production is consistent with quantum theory.

From the modern point of view,
thermalization of isolated quantum systems has been of much interest
in light of state-of-the-art technologies such as ultracold atoms
\cite{Kinoshita2006,Hofferberth2007,Trotzky2011,Gring2012,Langen2013,Kaufman2016},
trapped ions~\cite{Clos2016}, and superconducting qubits~\cite{Neill2016}.
For example,
it has been established that even a pure quantum state in an isolated system exhibits relaxation toward thermal equilibrium by  unitary dynamics
\cite{VonNeumann1929,Jensen1985,Tasaki1998,Rigol2008,Reimann2008,Linden2009,Goldstein2010,Reimann2010,Riera2012,Mueller2013,Goldstein2015,Reimann2015,Tasaki2015,Farrelly2016,Gogolin2016}. 
Several key theoretical ideas have been developed to understand such relaxation processes.
It has been rigorously proved that an overwhelming majority of pure states in the Hilbert space describe thermal equilibrium,
which is referred to as \textit{typicality}~\cite{Popescu2006,Goldstein2006,Sugita2006a,Reimann2007}.
The eigenstate thermalization hypothesis (ETH) 
is another cornerstone to understand thermalization~\cite{Jensen1985,Rigol2008,Berry1977,Peres1984,Deutsch1991,Srednicki1994,Biroli2010,Kim2014,Beugeling2014,Garrison2015},
which states that \textit{all} the energy eigenstates in an energy shell are thermal.
The ETH (or the \textit{strong} ETH) is believed to be true for non-integrable systems without disorder.
We note that there is a weaker version of the ETH (the \textit{weak} ETH)~\cite{Biroli2010}, which states that \textit{most} of the energy eigenstate are thermal.
The weak ETH is rigorously proved for both integrable and non-integrable systems without disorder~\cite{Mori2016,Iyoda2016}.

On the other hand, 
research into the second law,
such as the fluctuation theorem~\cite{Jarzynski1997,Crooks1999,Jarzynski1999,Kurchan2000,Tasaki2000,Esposito2009,Campisi2011,Sagawa2012,Alhambra2016,Morikuni2016}
and thermodynamic resource theories~\cite{Horodecki2013,Brandao2015},
is based on a special assumption that the initial state of a heat bath is the canonical ensemble.
This assumption is very strong,
given that a thermal equilibrium state can even be a pure quantum state.
If we do not assume the initial canonical distribution,
the non-negativity of the entropy production cannot be proved by
the conventional arguments,
i.e.,
the fluctuation theorem or the positivity of relative entropy~\cite{Sagawa2012}.
In this sense, 
except for only a few seminal works~\cite{Tasaki2000a,Goldstein2013,Ikeda2013},
the foundation of the second law in isolated quantum systems has not been fully addressed.

In this paper, we rigorously prove that the entropy production saturates at some value in the long time regime,
for a broad class of initial states including pure states.
We note that
\textit{the entropy production,}
the sum of the change in von Neumann entropy of a subsystem and the heat absorbed from a heat bath,
can change in time,
though the von Neumann entropy of the total system is invariant under unitary dynamics~\cite{nielsen2000quantum}.
In general,
the saturation value of the entropy production depends on both the initial state and the Hamiltonian, and could be negative.
We then prove that the saturation value is non-negative
if the system satisfies the strong ETH,
which leads to the second law of thermodynamics in the long time regime.
Our results provide
a quantum foundation to observation by Clausius.

\begin{figure}
	\includegraphics[width=0.9\linewidth]{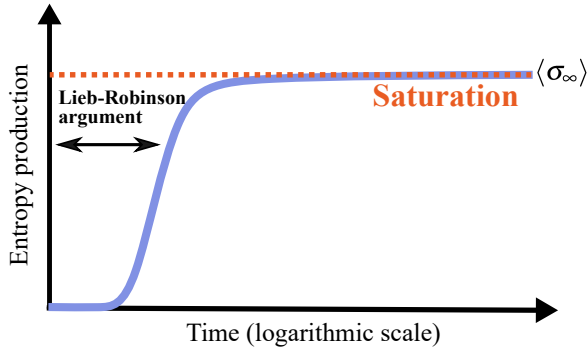}
	\caption{\label{fig:schematic}
		(color online)
		Schematic of a time evolution of the entropy production.
		Starting from a non-equilibrium initial state, the entropy production increases, and then saturates at  some value $ \epde $.
		The long time behavior is the main focus of the present work,
		while the short time behavior can be explained by our previous work~\cite{Iyoda2016} based on the Lieb-Robinson bound.}
\end{figure}

We note that in our previous work~\cite{Iyoda2016},
we have proved the second law for pure quantum states in the short time regime,
on the basis of the Lieb-Robinson bound~\cite{Lieb1972,Nachtergaele2006,Hastings2006}.
However, the previous result cannot apply to the long time regime.
In the present work,
we have adopted a completely different technique from the Lieb-Robinson bound,
and have proved the second law in the long time regime (see Fig.~\ref{fig:schematic}).

The rest of this paper is organized as follows.
In Sec.~\ref{sec:setup}, we formulate our setup and 
make a brief review on equilibration in isolated quantum systems.
In Sec.~\ref{sec:saturation}, we prove the main theorem, which states that the entropy production takes a constant value for most times.
In Sec.~\ref{sec:Non-negativity}, we derive a lower bound of the saturation value of the entropy production.
We also show that the saturation value is non-negative if the system satisfies the ETH.
In Sec.~\ref{sec:Numerical}, we confirm our theory by numerical simulation of hard-core bosons.
In Sec.~\ref{sec:Conclusion}, we summarize our results.
In Appendix~\ref{sec:effective_dimension}, we make a remark on the effective dimension.

\section{\label{sec:setup}Setup and notation}
We consider a quantum many-body system described by a finite-dimensional Hilbert space.
The total system is divided into two parts,
a small system $ \rmS $ and a large heat bath $ \rmB $.
Their Hilbert space dimensions are denoted by $ D_\rmS $ and $ D_\rmB $, respectively.
The total Hamiltonian is given by
\begin{align}
H=\hs+\hb+\hi,
\end{align}
where $ \hs $ is the Hamiltonian of  system $ \rmS $,
$ \hb $ is the Hamiltonian of  bath $ \rmB $,
and $ \hi $ is the interaction between system $ \rmS $ and bath $ \rmB $.
We write the spectrum decomposition of $ H $ as
\begin{align}
H= \sum_j E_jP_j,
\end{align}
where $ E_j $ is an eigen-energy ($ E_i\neq E_j $ for $ i\neq j $),
and $ P_j $ is the projector onto the eigenspace of $ H $ with energy eigenvalue $ E_j $.

\begin{figure}
	\includegraphics[width=0.9\linewidth]{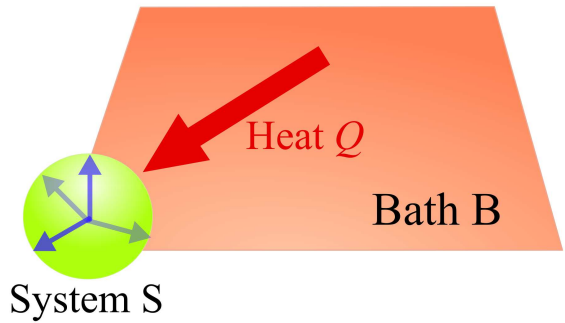}
	\caption{\label{fig:setup}
		(color online)
		Schematic of our setup.
		The total system consists of system $ \rmS $ and bath $ \rmB $.}
\end{figure}

We assume that 
the initial correlation between system $ \rmS $ and bath $ \rmB $ is zero.
The initial density operator of the total system is then given by
\begin{align}
\rho(0)=\rs(0)\otimes\rb(0),
\label{eq:initial_state}
\end{align}
where $ \rs(0) $ and $ \rb(0) $ are respectively
the initial density operators of system $ \rmS $ and bath $ \rmB $.
We do not assume that $ \rb(0) $ is the canonical ensemble.

The total system obeys the unitary time evolution generated by the total Hamiltonian $ H $.
The state at time $ t $ is written by $ \rho(t)=e^{-iH t}\rho(0)e^{iH t} $.
Note that we set $ \hbar=1 $.
We also write the time-averaged state by
\begin{align}
\rde:=\lim_{\tau\to\infty}\frac{1}{\tau}\int_0^{\tau}\rho(t)dt,
\label{eq:def_DE}
\end{align}
which is called the diagonal ensemble.
Because of the recurrence property of the unitary evolution, 
an isolated quantum system cannot converge to an exact stationary state.
However, a system can relax to an effective stationary state most of the time between recurrences.
In fact, it has been proved in Refs.~\cite{Linden2009,Short2012} that for any operator $ O $,
and for any $ \epsilon>0 $,
\begin{align}
\lim_{\tau\to\infty}\prob_{t\in[0,\tau]}\left[\abs{\tr[O\rho(t)]-\tr[O\rde]}\geq\epsilon\onorm{O}\right]
\leq\frac{\ndg}{\epsilon^2\deff}.
\label{eq:equilibration1}
\end{align}
Here, $ \prob_{t\in[0,\tau]} $ is the uniform measure over $ t\in [0,\tau]  $, and
$ \onorm{O} $ is the operator norm of $ O $.
$ \ndg $ is the degeneracy of the most degenerate energy gap,
which is given by $ \ndg :=\max_{j\neq k}\abs{\{(l,m):G_{(l,m)}=G_{(j,k)}\}} $
with $ G_{(j,k)}:=E_j-E_k $.
If the total Hamiltonian $ H $ has  no degenerate energy gap, $ \ndg=1 $.
$ \deff $ is the effective dimension of the initial state $ \rho(0) $,
which is defined as $ \deff:=1/\sum_j(\tr[P_j\rho(0)])^2 $.
If the Hamiltonian has no degeneracy,
the effective dimension is equal to the inverse of the purity of the diagonal ensemble, i.e., $ 1/\tr[\rde^2] $.
Inequality~\eqref{eq:equilibration1} means that observables equilibrate to their diagonal ensemble averages
if $ \deff $ is sufficiently large.

To quantify the distance between two density operators,
we use the trace distance $ \trdist{\rho}{\tau}:=\frac{1}{2}\tr[\sqrt{(\rho-\tau)^2}] $.
When it is small, we hardly distinguish two states.
We also use the local distance $ \ldist{\rho}{\tau}:=\trdist{\rs}{\tau_{\rmS}}$,
where we write $ \rs:=\tr_\rmB[\rho]$ for the reduced density operator.
We regard that system $ \rmS $ equilibrates to the diagonal ensemble 
when  the local distance $ \ldist{\rho(t)}{\omega}$ is small most of the time.
It has been shown  in Refs.~\cite{Linden2009,Short2012}  that for any positive number $ \epsilon>0 $,
\begin{align}
\lim_{\tau\to\infty}\prob_{t\in[0,\tau]}\left[\ldist{\rho(t)}{\omega}\geq\epsilon\right]
\leq\frac{1}{2}\sqrt{\frac{D_{\rmS}^2\ndg}{\epsilon^2\deff}}.
\label{eq:equilibration2}
\end{align}
This inequality implies that system $ \rmS $ equilibrates to the diagonal ensemble
if $ \deff $ is sufficiently large compared to the Hilbert space dimension of system $ \rmS $.

The effective dimension $ \deff $ of the initial product state $ \rho(0) $ is bounded from below
by the purity of the initial state of bath $ \rmB $:
\begin{align}
\deff \geq \frac{1}{\nde \tr[\rb(0)^2]},
\label{eq:bound_eff}
\end{align}
where $ \nde $ is the maximum energy degeneracy of the total Hamiltonian $ H $.
The proof of inequality \eqref{eq:bound_eff} is the following.
We denote by $ K_j $ an orthonormal basis of the degenerate eigenspace of energy  $ E_j $.
Then, we have
\begin{align}
\frac{1}{\deff}
&=\sum_j(\tr[P_j\rs(0)\otimes\rb(0)])^2
\\
&=
\sum_j\left(\sum_{\ket{E_k}\in K_j}\braket{E_k|\rs(0)\otimes\rb(0)|E_k}\right)^2
\\
&\leq
\sum_j|K_j|\sum_{\ket{E_k}\in K_j}\braket{E_k|\rs(0)\otimes\rb(0)|E_k}^2
\\
&\leq \nde\sum_j
\sum_{\ket{E_k}\in K_j}\braket{E_k|\rs(0)\otimes\rb(0)|E_k}^2
\\
&\leq
\nde\tr[(\rs(0)\otimes\rb(0))^2]
\\
&\leq
\nde\tr[\rb(0)^2],
\end{align}
where we used $ \tr[\rs(0)^2]\leq 1 $ to obtain the last line.

When the initial state of bath $ \rmB $ is the microcanonical ensemble $ \rbmc:=P_{\rmB}^{\mrm{mc}}/\tr[P_{\rmB}^{\mrm{mc}}]  $,
where $ P_{\rmB}^{\mrm{mc}} $ is the projector onto the Hilbert space of the energy shell  $M_{U,\Delta}= \{E^{\rmB}_j:\abs{E^{\rmB}_j-U}\leq \Delta\} $
with $ \Delta >0 $ (see Sec.~\ref{subsec:ETH} for more details),
or the canonical ensemble $ \rbcan:=e^{-\beta\hb}/\tr[e^{-\beta\hb}] $,
then the purity $ \tr[\rb(0)^2] $ is exponentially small in the size of bath $ \rmB $
(see Appendix~\ref{sec:effective_dimension}).
Thus $ \deff $  is exponentially large in the size of bath $ \rmB $.
Furthermore, when the initial state of bath $ \rmB $ is randomly taken from the Hilbert space of the energy shell according to the Haar measure,
$ \deff $ typically becomes exponentially large in the size of bath $ \rmB $
(see again Appendix~\ref{sec:effective_dimension}).
Therefore, system $ \rmS $ equilibrates to the diagonal ensemble
if the initial state of bath $ \rmB $ is the microcanonical ensemble, the canonical ensemble, or a typical pure state,
given that bath $ \rmB$ is sufficiently large.

Finally, we define the entropy production from time $ 0 $ to $ t $ by
\begin{align}
\braket{\sigma(t)}=\Delta \ssy(t)-\beta Q(t),
\end{align}
where $ \Delta \ssy(t):=\ssy(t)-\ssy(0)$ is the change in the von Neumann entropy of system $ \rmS $ with $ \ssy(t):=S(\rs(t))=-\tr_{\rmS}[\rs(t)\log\rs(t)] $,
and $ Q(t):=-\tr_{\rmB}[\hb(\rb(t)-\rb(0)]  $ is the heat absorbed by system $ \rmS $.
If the initial state of bath $ \rmB $ is the canonical ensemble $ \rb(0)=\rbcan $,
the entropy production $ \braket{\sigma(t)} $ is always non-negative: $ \braket{\sigma(t)}\geq 0 $~\cite{Sagawa2012},
which is an expression of the second law of thermodynamics.
This inequality results from the non-negativity of the quantum relative entropy.
When the initial state of bath $ \rmB $ is not the canonical ensemble,
the entropy production $ \braket{\sigma(t)} $ is not necessarily non-negative.
In this paper, we discuss the behavior of the entropy production and the validity of the second law for general initial states.

\section{\label{sec:saturation}Saturation of the entropy production}
In this section, 
we discuss our main topic:
saturation of the entropy production under reversible unitary dynamics.
In Sec.~\ref{sec:setup}, 
we discussed that system $ \rmS $ equilibrates to the diagonal ensemble if the effective dimension is sufficiently large,
i.e., inequalities~\eqref{eq:equilibration1}~and~\eqref{eq:equilibration2}.
However, the entropy production cannot be described as the expectation value of any instantaneous observable.
Therefore, we cannot directly apply the argument in Sec.~\ref{sec:setup} to the entropy production.
Instead, we evaluate the difference between the entropy production at time $ t $ 
and the saturated entropy production $ \epde $ by using inequalities~\eqref{eq:equilibration1}~and~\eqref{eq:equilibration2}.

We now formulate our theorem rigorously.
We assume that the initial state of the total system is given by a product state \eqref{eq:initial_state}.
The saturated entropy production $ \epde $ is defined as
\begin{align}
\epde :=S(\rde_{\rmS})-S(\rs(0))-\beta \tr[\hb(\rho(0)-\rde)],
\label{eq:saturation_value}
\end{align}
where $ \rde $ is defined in Eq.~\eqref{eq:def_DE},
and $ \rde_{\rmS}:=\tr_{\mrm{B}}[\rde] $.
If the difference $ \abs{\braket{\sigma(t)}-\epde} $ is small most of the time,
the entropy production saturates at $ \epde $.
In fact, we show the following theorem.

\begin{theorem}[\label{thm:saturation}Saturation of the entropy production]
For any $ 0<\epsilon<1/2 $,
\begin{align}
&\lim_{\tau\to\infty}\prob_{t\in[0,\tau]}\Bigl[
\abs{\braket{\sigma(t)}-\epde}
\nonumber\\
&\qquad\leq
\epsilon\log (D_{\rmS}-1)+H_2(\epsilon)
+\epsilon\beta\onorm{\hs+\hi}
\Bigr]
\nonumber\\
&\span \geq
1-\frac{1}{2}\sqrt{\frac{D_{\rmS}^2\ndg}{\epsilon^2\deff}}-\frac{\ndg}{{\epsilon}^2\deff}
\label{eq:saturation},
\end{align}
where $ H_2(x):=-x\log x-(1-x)\log(1-x) $.
\end{theorem}

\begin{proof}
From the definitions of $ \braket{\sigma(t)} $ and $ \epde $,
we have
\begin{align}
&\abs{\braket{\sigma(t)}-\epde}
\nonumber\\
& \leq \abs{S(\rs(t))-S(\rde_{\rmS})}
+\beta\abs{\tr\left(\hb\rho(t)-\rde\right)}.
\end{align}
By applying the Fannes-Audenaert inequality~\cite{Audenaert2006}, the first term on the right-hand side is bounded by
\begin{align}
&\abs{S(\rs(t))-S(\rde_{\rmS})}
\nonumber\\
& \leq \mathcal{D}_{\rmS}(\rho(t),\rde)\log (D_{\rmS}-1)+H_2(\mathcal{D}_{\rmS}(\rho(t),\rde)).
\end{align}
Using inequality~\eqref{eq:equilibration2},
we evaluate $ \mathcal{D}_{\rmS}(\rho(t),\rde) $,
and then we have for any $ 0<\epsilon<1/2 $,
\begin{widetext}
\begin{align}
\lim_{\tau\to\infty}\prob_{t\in[0,\tau]}\left[
\abs{S(\rs(t))-S(\rde_{\rmS})}
\leq \epsilon\log (D_{\rmS}-1)+H_2(\epsilon)
\right]
\geq
1-\frac{1}{2}\sqrt{\frac{D_{\rmS}^2\ndg}{\epsilon^2\deff}}.
\label{eq:entropy}
\end{align}
\end{widetext}
Since the expectation value of the total Hamiltonian $H$ is the same for both the state $ \rho(t) $ and the diagonal ensemble $ \rde $,
we have
\begin{align}
\tr\left(\hb(\rho(t)-\rde)\right)=-\tr\left[(\hs+\hi)(\rho(t)-\rde)\right].
\end{align}
By using inequality~\eqref{eq:equilibration1},
we have for any $ 0<\epsilon<1/2 $,
\begin{align}
&\lim_{\tau\to\infty}\prob_{t\in[0,\tau]}\left[
\beta\abs{\tr\left(\hb(\rho(t)-\rde)\right)}
\leq \epsilon\beta\onorm{\hs+\hi}
\right]
\nonumber\\
&\qquad
\geq
1-\frac{\ndg}{{\epsilon}^2\deff}.
\label{eq:heat}
\end{align}
By combining inequality \eqref{eq:entropy} with \eqref{eq:heat},
we finally obtain inequality \eqref{eq:saturation}.
\end{proof}

We have made no assumption on the initial state of bath $ \mrm{B} $;
Theorem \ref{thm:saturation} is applicable not only to the canonical ensemble, but also to an arbitrary initial state.
From Theorem~\ref{thm:saturation}
we find that a sufficient condition of  the saturation is just the large effective dimension of the initial state.
As we remarked in Sec.~\ref{sec:setup}, the effective dimension is exponentially large with the size of bath $ \mrm{B} $,
when the initial state of bath $ \mrm{B} $ is the microcanonical ensemble, 
the canonical ensemble,
or a typical quantum state.
From Theorem~\ref{thm:saturation},
the entropy production saturates at $ \epde $ for these cases.

We remark that
the time scale,
which makes the long time limit in Theorem~\ref{thm:saturation} meaningful,
is expected to be doubly exponentially large with respect to the size of the total system
for non-integrable systems.
This is much longer than the time scale determined by the Lieb-Robinson bound.
However, 
it has been observed numerically and analytically that
the informational quantities, such as the entanglement entropy,
ballistically spread and saturate at around the Lieb-Robinson time
for both of integrable and  non-integrable systems ~\cite{calabrese2005,DeChiara2006,Lauchli2008,Kim2013}.
As a result,
the entropy production,
defined with the von Neumann entropy,
is expected to relax to the long-time average at around the Lieb-Robinson time,
as numerically observed in Ref.~\cite{Iyoda2016}
and also in Sec.~\ref{sec:Numerical}.
Therefore, we expect that in practical situations,
there is no intermediate time regime between the saturation of the entropy production and the Lieb-Robinson time.
However,
for the moment,
it seems to be formidably difficult to prove it rigorously for general situations.

As mentioned before, it has been shown that the entropy production is always non-negative
when the initial state of bath $ \mrm{B} $ is the canonical ensemble~\cite{Sagawa2012}.
In this case, the entropy production increases and saturates at 
a non-negative value,
which is consistent with the second law of thermodynamics.

On the other hand,
the saturation value $ \epde $ is not necessarily non-negative
for a general initial state. 
In fact, $ \epde $ is  determined by the diagonal ensemble $ \rde $,
which in general depends on the details of the initial state.
For this reason, it is difficult to prove the  non-negativity of $ \epde $ for completely general situations.

If the total system satisfies the strong ETH,
$ \rde_{\rmS} $ does not depend 
on the details of the initial state,
but only depends on the initial average energy.
We note that the strong ETH is believed to be true for non-integrable systems~\cite{Rigol2008,Biroli2010,Kim2014,Beugeling2014,Garrison2015}
(see also~\cite{Gogolin2016} for definitions of integrability in quantum systems).
In such a case,
the diagonal ensemble $ \rde_{\rmS} $  of a general initial state
is the same as that of the initial canonical ensemble,
and therefore $ \epde $ takes the same value as the canonical case,
where $ \epde $ is non-negative.
This is the scenario that we will prove in the next section by assuming the strong ETH.

\section{\label{sec:Non-negativity}Non-negativity of the saturated entropy production}

We now discuss the non-negativity of the saturated entropy $ \epde $.
Specifically, we utilize the ETH~\cite{Jensen1985,Rigol2008,Berry1977,Peres1984,Deutsch1991,Srednicki1994,Biroli2010,Kim2014,Beugeling2014,Garrison2015},
which is sufficient to show the non-negativity of  $ \epde $.

In Sec.~\ref{subsec:bound}, we prove a general lower bound of $ \epde $
without using the ETH.
Then, in Sec.~\ref{subsec:ETH}, we show that the lower bound approaches to zero in the large bath limit,
if the system satisfies the strong ETH.
In Sec.~\ref{subsec:weak ETH},
we discuss consequences of the weak ETH.

\subsection{\label{subsec:bound}Lower bound of the saturation value}
First, we derive a lower bound of the saturated entropy production
by comparing the saturation value for a general initial state of the bath, $ \rb(0) $, 
with that for the canonical initial state, $ \rbcan $.
For this purpose, we choose the canonical ensemble $ \rbcan $ such that $ \tr_{\mrm{B}}[\hb\rbcan]= \tr_{\mrm{B}}[\hb\rb(0)]$ for a given $ \rb(0) $.
We denote the diagonal ensemble of $ \rs(0)\otimes\rbcan $ by $ \rcande $,
and the local trace distance between $ \rde $ and $ \rcande $ by $ \dmccan:=\mathcal{D}_{\mrm{S}}(\rde,\rcande) $.

Let us first remark on a special case
where the deviation vanishes as $ \dmccan\to0 $ in the thermodynamic limit,
and
the interaction energy is negligible.
In such a case,
the system equilibrates to the same state
as the case that the bath is initially in the canonical ensemble $ \rbcan $.
In this case, the saturation value of the entropy production for a general initial state takes the same value as the canonical initial state,
and therefore it is non-negative.
In Sec.~\ref{subsec:ETH},
we will show that the deviation indeed vanishes, $ \dmccan\to0 $,
when the strong ETH holds for the total system.

For general quantum systems,
the deviation does not necessarily vanish even in the thermodynamic limit.
By keeping $ \dmccan $ finite,
and also by taking the interaction energy into consideration,
we have the following lower bound of $ \epde $.

\begin{theorem}[\label{thm:lower_bound}A lower bound of the saturation value]
The saturated entropy production $ \epde $ satisfies
\begin{align}
\epde\geq -\dep.
\label{eq:lower_bound}
\end{align}
Here, we defined the error term $ \dep $ as 
\begin{align}
\delta_{\mrm{EP}}:=\dmccan\log (D_{\rmS}-1)+H_2(\dmccan)
+2\dmccan\beta\onorm{\hs}
+\beta\dei,
\label{eq:def_dep}
\end{align}
where $ \dei $ is the change in the interaction energy defined as
\begin{align}
\dei:&=\left|\tr\left[\hi(\rs(0)\otimes\rb(0)-\rde)\right]\right.
\nonumber\\
&\span
\left.-\tr\left[\hi(\rs(0)\otimes\rbcan-\rcande)\right]\right|.
\label{eq:def_dei}
\end{align}
\end{theorem}

\begin{proof}
The entropy production with the initial canonical state $ \rs(0)\otimes\rbcan $, corresponding to the diagonal ensemble $ \rcande $,
is given by
\begin{align}
\epcande&:=S(\rcande_{\mrm{S}})-S(\rs(0))-\beta\tr\left[\hb(\rbcan-\rcande)\right].
\end{align}
The deviation between $ \epde $ and $ \epcande $   is
\begin{align}
&\epde-\epcande
\nonumber\\
&=
S(\rde_{\mrm{S}})-S(\rcande_{\mrm{S}})
\nonumber\\
&\qquad
+\beta\left(\tr\left[\hb(\rde_{\mrm{B}}-\rb(0))\right]-\tr\left[\hb(\rcande_{\mrm{B}}-\rbcan)\right]\right).
\label{eq:dev_EP}
\end{align}
By applying the Fannes-Audenaert inequality,
the difference of the von Neumann entropies on the right-hand side of Eq.~\eqref{eq:dev_EP} is evaluated as
\begin{align}
\abs{S(\rde_{\mrm{S}})-S(\rcande_{\mrm{S}})}\leq \dmccan\log(D_{\rmS}-1)+H_2(\dmccan).
\label{eq:dev_NE_bound}
\end{align}
The total average energy is the same for both the initial state and the diagonal ensemble.
Therefore, we have
\begin{align}
&\tr\left[\hb(\rde-\rs(0)\otimes\rb(0))\right]
\nonumber\\
&\qquad=-\tr\left[(\hs+\hi)(\rde-\rs(0)\otimes\rb(0))\right],
\label{eq:energy_consv_mc}
\\
&\tr\left[\hb(\rcande-\rs(0)\otimes\rbcan)\right]
\nonumber\\
&\qquad=-\tr\left[(\hs+\hi)(\rcande-\rs(0)\otimes\rbcan)\right].
\label{eq:energy_consv_can}
\end{align}
From  inequalities~\eqref{eq:energy_consv_mc} and \eqref{eq:energy_consv_can},
we can evaluate $ \tr\left[\hb(\rde-\rcande)\right] $ as
\begin{align}
\abs{\tr\left[\hb(\rde-\rcande)\right]}
\leq
2\onorm{\hs}\dmccan+\dei.
\label{eq:dev_HB_bound}
\end{align}
From inequalities~\eqref{eq:dev_HB_bound} and \eqref{eq:dev_NE_bound},
we have
\begin{align}
&\abs{\epde-\epcande}\nonumber\\
&\leq
\dmccan\log(D_{\rmS}-1)+H_2(\dmccan)+2\beta\onorm{\hs}\dmccan+\beta\dei.
\label{eq:dev_EP_bound}
\end{align}
By combining inequality \eqref{eq:dev_EP_bound} and $ \epcande\geq0 $,
we prove the theorem.
\end{proof}

The first three terms on the right-hand side of Eq.~\eqref{eq:def_dep} vanishes in the limit of $ \dmccan\to 0 $.
The last term on the right-hand side of Eq.~\eqref{eq:def_dep} vanishes (i.e., $ \dei=0 $) if 
\begin{align}
[\hs+\hb,\hi]=0 ,
\label{eq:hs+hb_consv}
\end{align}
which means that the sum of the energies of system $ \mrm{S} $ and bath $ \mrm{B} $ is a conserved quantity,
and the interaction energy does not change during the dynamics. 
The same condition as Eq.~\eqref{eq:hs+hb_consv} is also assumed in the thermodynamic resource theory~\cite{Horodecki2013,Brandao2015}.

We note that Eq.~\eqref{eq:hs+hb_consv} does not necessary imply that
the interaction itself is weak.
For example, the Jaynes-Cummings model with the resonance condition satisfies Eq.~\eqref{eq:hs+hb_consv}
even when the interaction is strong~\cite{Jaynes1963}.
Although a general quantum system does not necessarily satisfy the condition \eqref{eq:hs+hb_consv},
we have numerically shown that the contribution from
$\dei$ is negligibly small in some non-integrable models (see Sec.~\ref{sec:Numerical} and Ref.~\cite{Iyoda2016}).
In general, however,
it is still unclear under what condition the interaction term $ \dei $
is negligible without weak coupling limit,
which is a future issue.

\subsection{\label{subsec:ETH}Non-negativity from the strong ETH}

We next discuss the initial state independence (i.e., $ \dmccan\to0 $)
on the basis of the strong ETH,
which states that \textit{every} energy eigenstate represents  thermal equilibrium.
There has not yet been a mathematical proof of the strong ETH,
while it has been confirmed by a lot of numerical experiments~\cite{Rigol2008,Biroli2010,Kim2014,Beugeling2014,Garrison2015}.
In this section, we just assume the strong ETH as a starting point of our argument.

To state the strong ETH rigorously,
we focus on the situation that the total system has
the well-defined thermodynamic limit,
when only bath $ \rmB $ becomes large and
system $ \rmS $ is kept small.
Let $ N $ be the total system size (i.e., the number of lattice sites or particles).
For simplicity, we also assume that Hamiltonian $ H $ has no degeneracy (i.e., $ \nde=1 $).

The energy shell with energy density $ u $ and width $ \Delta $ is defined as
the following set of indexes of energy eigenvalues of the total Hamiltonian $ H $:
\begin{align}
M_{uN,\Delta}:=\{j:\abs{E_j-uN}\leq\Delta\}.
\label{eq:energy_shell}
\end{align}
Here, 
$ \Delta >0 $ can be arbitrary taken within the order of $ \Delta=\mathcal{O}(N^\alpha) $ with $ 0\leq\alpha<1 $.
The microcanonical ensemble of the energy shell $ M_{uN,\Delta} $ is denoted by $ \rmc $.
We note that in rigorous statistical mechanics~\cite{Ruelle},
the energy shell is defined as
$ M'_{U',\Delta'}:=\{j:U'-\Delta'\leq E_j\leq U'\} $,
where $ \lim_{N\to\infty}U'/N $ exists and is finite,
and $ \Delta'=\mathcal{O}(N^{\alpha'}) $ with  $ 0\leq\alpha'\leq1 $.
With this definition, many of the fundamental properties in standard statistical mechanics, such as Boltzmann's formula and the equivalence of ensembles,
have been rigorously proved~\cite{Ruelle,Tasaki2016}.
Our definition of the energy shell \eqref{eq:energy_shell} is regarded as a special case of the above standard definition,
by taking $ U'=uN+\Delta $ and $ \Delta'=2\Delta $ 
(i.e.,
$ M_{uN,\Delta}= M'_{uN+\Delta,2\Delta }$).

In the present context, the strong ETH of the total system is defined as follows.

\begin{assumption}[\label{dfn:Strong-ETH}Strong ETH]
We say that the total system satisfies the strong ETH,
if for any $ \epsilon>0 $,
\begin{align}
\ldist{\ket{E_j}\bra{E_j}}{\rmc}<\epsilon
\quad\text{for any}\  j\in M_{uN,\Delta}
\label{eq:strong ETH}
\end{align}
holds for sufficiently large $ N $.
\end{assumption}

From the non-degeneracy of the energy eigenvalues, 
the diagonal ensemble can be represented as
\begin{align}
	\rde=\sum_j\ket{E_j}\braket{E_j|\rs(0)\otimes\rb(0)|E_j}\bra{E_j}
	\label{eq:DE_non-deg}.
\end{align}
As shown below,
by using the strong ETH,
we can show the initial state independence of $ \dmccan\to0  $ in the thermodynamic limit.
To apply the strong ETH, 
we need to make an additional assumption on the initial energy distribution.

\begin{assumption}[\label{asmp:energy_distribution}Energy distribution of the initial state]
Let $ P_{\mrm{out}} $ be the projection operator onto the subspace corresponding to the outside of the energy shell:
$  P_{\mrm{out}}:=\sum_{j\not\in M_{uN,\Delta}} \ket{E_j}\bra{E_j}$.
We assume that for any $ \epsilon>0 $, 
\begin{align}
\tr[P_{\mrm{out}}\rs(0)\otimes\rb(0)]&<\epsilon
\label{eq:assumption1}
\end{align}
holds for sufficiently large $N $.
\end{assumption}

This assumption implies that the energy distribution of the initial state is narrow
and centered around $ uN $.
This is a reasonable
assumption if we take the width of the energy shell 
within the order of $ \Delta =\mathcal{O}(N^\alpha) $ with $ 1/2<\alpha<1 $.
In fact,
for the case of the canonical ensemble,
the standard deviation of the energy distribution
is the order of $ \mathcal{O}(N^{1/2}) $~\cite{Tasaki2016}.

If the energies of the system and the interaction are small enough than the energy of  the bath,
the energy distribution of the total initial state is similar to that of the bath.
We can thus expect that Assumption \ref{asmp:energy_distribution} would hold
when the bath is initially in the microcanonical ensemble, the canonical ensemble, or a pure state in the energy shell.

We now evaluate the local trace distance between $ \rde $ and $ \rmc $:

\begin{lemma}[\label{lem:thermalization}Initial state independence from the strong ETH]
Under Assumptions \ref{dfn:Strong-ETH} and \ref{asmp:energy_distribution},
for any $ \epsilon>0 $,
\begin{align}
\mathcal{D}_{\mrm{S}}(\rde,\rmc)<\epsilon
\label{eq:theramaliztion}
\end{align}
holds for sufficiently large $ N $.
\end{lemma}

\begin{proof}
	From the triangle inequality, we find
\begin{align}
	\ldist{\rde}{\rmc}
	\leq
	\ldist{\rde}{\tilde{\rde}}
	+
	\ldist{\tilde{\rde}}{\rmc},
\end{align}
where we defined the diagonal ensemble projected onto the energy shell by
$ \tilde{\rde}:=(1- P_{\mrm{out}})\rde(1- P_{\mrm{out}})/\tr[(1- P_{\mrm{out}})\rde] $.
From Assumption \ref{dfn:Strong-ETH},
for any $ \epsilon>0 $,
\begin{align}
\ldist{\tilde{\rde}}{\rmc}
<\frac{\epsilon}{2}
\end{align}
holds for sufficiently large $ N $.
On the other hand, from Assumption \ref{asmp:energy_distribution},
for any $ \epsilon>0 $,
\begin{align}
\ldist{\rde}{\tilde{\rde}}<\frac{\epsilon}{2}
\end{align}
holds for sufficiently large $ N $.
By summing up these two terms,
we obtain Eq.~\eqref{eq:theramaliztion}.
\end{proof}

This result means that the system reaches an equilibrium state independent of the initial state,
if Assumptions \ref{dfn:Strong-ETH} and \ref{asmp:energy_distribution} are satisfied. 
Then, by combining Lemma~\ref{lem:thermalization} with Theorem \ref{thm:lower_bound},
we obtain the non-negativity of the saturation value in the thermodynamic limit.

\begin{theorem}[\label{thm:non-negativity}Non-negativity from the strong ETH]
We assume that $ [\hs+\hb,\hi]=0 $,  along with Assumptions \ref{dfn:Strong-ETH} and \ref{asmp:energy_distribution}.
For any $ \epsilon>0 $,
\begin{align}
\epde\geq-\epsilon
\label{eq:non-negativity}
\end{align}
holds for sufficiently large $ N $.
\end{theorem}

\begin{proof}
In Eq.~\eqref{eq:def_dep},
$ \dmccan $ vanishes from Lemma~\ref{lem:thermalization},
and $ \dei $ vanishes from $ [\hs+\hb,\hi]=0 $.
\end{proof}

We note that Theorem~\ref{thm:non-negativity} is valid as long as its assumptions are satisfied,
regardless of the initial state of the bath.
When commutator $ [\hs+\hb,\hi] $ is not equal to $ 0 $,
an negative error term $ -\beta\dei $ defined in Eq.~\eqref{eq:def_dei}
is just added  to the right-hand side of inequality \eqref{eq:non-negativity}.
However,  $ \beta\dei $ is expected to be small for some cases,
as will be numerically demonstrated in Sec.~\ref{sec:Numerical}.

We have not mentioned the size dependence of the error term in inequality~\eqref{eq:non-negativity},
which generally depends on the nontrivial scaling in 
inequalities \eqref{eq:strong ETH} and \eqref{eq:assumption1}.
We will numerically show in Sec.~\ref{sec:Numerical} that $ \epde $ becomes non-negative without error,
even with a bath of $ 16 $ sites.

By combining the saturation of the entropy production (Theorem \ref{thm:saturation}) with the non-negativity (Theorem \ref{thm:non-negativity}),
we find that 
\begin{align}
 \braket{\sigma(t)}\gtrsim 0
\end{align}
holds for a typical time $ t $.
This inequality represents the second law of thermodynamics at late times
for general initial states.
We note that our theory implies a stronger statement than just the non-negativity of the entropy production:
it \textit{saturates} at a non-negative value.

\subsection{\label{subsec:weak ETH}Non-negativity from the weak ETH}

Non-negativity of the saturated entropy production can also be shown by using the weak ETH~\cite{Biroli2010,Iyoda2016,Mori2016} instead of the strong ETH.
In fact, if the initial state has a sufficiently large effective dimension,
we can prove the initial-state independence of $ \rde_{\rmS} $.

In the present context,
we formulate the weak ETH as the following large-deviation type bound.

\begin{assumption}[\label{asmp:weak-ETH}Weak ETH]
	We say that the total system satisfies the weak ETH,
	if for any $ \epsilon>0 $, there exists a constant $ \gamma_{\epsilon}>0 $
	such that
	\begin{align}
	\prob_{j\in M_{uN,\Delta}}[\ldist{\ket{E_j}\bra{E_j}}{\rmc}]> \epsilon]
	<e^{-\gamma_\epsilon N},
	\end{align}
	where $ \prob_{j\in M_{uN,\Delta}}[\cdots] $ represents the probability in the uniform distribution on the set of the energy eigenstates in the energy shell.
\end{assumption}
This form of the weak ETH has rigorously been proved for general lattice systems
with translation invariance,
including integrable systems~\cite{Tasaki2015,Mori2016}.
We note that the proof in Ref.~\cite{Tasaki2015,Mori2016} has been presented with
the choice of the energy shell
$ M'_{U',\Delta'}=\{j:U'-\Delta'\leq E_j\leq U'\} $ with $ \Delta'=\mathcal{O}(N) $,
while it is straightforward to generalize the proof to the case of 
our choice of the energy shell $ M_{uN,\Delta} $.
Under the weak ETH, a sufficient condition for the initial state independence is that
the initial state has an exponentially large effective dimension:

\begin{assumption}[\label{asmp:large_eff}Large effective dimension]
	We assume that for any $ \epsilon, \tilde{\epsilon}>0 $,
	the effective dimension of the initial state $ \rs(0)\otimes\rb(0) $ satisfies
	\begin{align}
	\deff>\frac{\dbmc e^{-\gamma_{\epsilon} N}}{\tilde{\epsilon}}
	\end{align}
	for sufficiently large $ N $.
\end{assumption}

Then, we evaluate the local trace distance between the diagonal ensemble $ \omega $ and the microcanonical ensemble $ \rmc $.

\begin{lemma}[\label{lem:thermalization_weak}Initial state independence from the weak ETH]
	Under Assumptions \ref{asmp:energy_distribution}, \ref{asmp:weak-ETH}, and \ref{asmp:large_eff},
	for any $ \epsilon>0 $,
	\begin{align}
	\ldist{\rde}{\rmc}<\epsilon
	\end{align}
	holds for sufficiently large $ N $.
\end{lemma}

\begin{proof}
From Assumption 2, $ \ldist{\rde}{\tilde{\rde}} $ vanishes in large $ N $.
	By using the Schwarz inequality and Assumption \ref{asmp:weak-ETH},
	we have for any $ \epsilon>0 $, 
	\begin{align}
	\ldist{\tilde{\rde}}{\rmc}
	<\epsilon+\sqrt{\frac{\dbmc e^{-\gamma_{\epsilon} N}}{\deff}}.
	\end{align}
	Therefore, by using Assumption~\ref{asmp:large_eff},
	we prove Lemma~\ref{lem:thermalization_weak}.
\end{proof}

Lemma~\ref{lem:thermalization_weak} means that the diagonal ensemble is locally thermal if the initial state has  the exponentially large effective dimension.
Then we can prove the non-negativity of the saturation value in the same manner as is the case for the strong ETH.

\begin{theorem}[\label{thm:non-negativity_weak}Non-negativity from the weak ETH]
	We assume that $ [\hs+\hb,\hi]=0 $, along with Assumptions \ref{asmp:energy_distribution}, \ref{asmp:weak-ETH}, and \ref{asmp:large_eff}.
	For any $ \epsilon>0 $,
	\begin{align}
	\epde\geq-\epsilon
	\end{align}
	holds for sufficiently large $N $.
\end{theorem}

We note that the validity of Assumption~\ref{asmp:large_eff} is highly nontrivial in practice.
On the other hand, an advantage of the argument by the weak ETH is that
it is applicable to both integrable and non-integrable systems.

\section{\label{sec:Numerical}Numerical simulation}
To confirm the validity of our theory,
we performed a numerical simulation by using exact diagonalization.
As a non-integrable model,
we consider a one-dimensional hard-core bosons with nearest-
and next-nearest-neighbor hoppings and interactions,
which satisfies the strong ETH~\cite{Kim2014}.
System $  \mrm{S} $ is a single site
and bath $ \mrm{B} $ consists of $ L $ sites.
The annihilation (creation) operator of a hard-core boson at site $ i $ is written as $ c_i $ ($ c_i^\dagger $),
which satisfies the commutation relations $ [c_i,c_j]=[c_i,c_j^\dagger]=[c_i^\dagger,c_j^\dagger]=0$ for $ i\neq j $,
$ \{c_i,c_i\}=\{c^\dagger_i,c_i^\dagger\}=0$, and $ \{c_i,c_i^\dagger\}=1 $.
The Hamiltonians are then given by
\begin{align}
\hs&=\omega c_0^\dagger c_0,\\
\hi&=-\gamma_{\mrm{I}}(c_0^\dagger c_1+c_1^\dagger c_0),\\
\hb&=\omega \sum_{i=1}^Lc^\dagger_ic_i\nonumber\\
&\quad+\sum_{i=1}^{L-1}[-\gamma(c^\dagger_ic_{i+1}+c_{i+1}^\dagger c_i )+gc^\dagger_ic_ic_{i+1}^\dagger c_{i+1}]
\nonumber\\
&\quad+\sum_{i=1}^{L-2}[-\gamma'(c^\dagger_ic_{i+2}+c_{i+2}^\dagger c_i )+g'c^\dagger_ic_ic_{i+2}^\dagger c_{i+2}]	,
\end{align}
where $ \omega>0 $ is the on-site potential, 
$ -\gamma_{\mrm{I}} $ is the hopping rate between system $ \mrm{S} $ and bath $ \mrm{B} $,
$ -\gamma $ and $ -\gamma' $ are the hopping rates in bath $ \mrm{B} $,
and $ g>0 $ and $ g'>0 $ are the repulsion energies.
We assume the open boundary condition.
We note that $ [\hs+\hb,\hi]\neq0 $ in this model.
We set the initial number of bosons in bath $ \rmB $ by $ N_{\mathrm{b}} $.
The initial state is a product state of system $ \rmS $ and bath $ \rmB $,
where the state of system $ \mrm{S} $ is given as $ \ket{1}:=c_0^\dagger\ket{0} $, and
the state of bath $ \rmB $ is an energy eigenstate $ \ket{E^{\rmB}_j} $.

Figure \ref{fig:effective_dimension} shows $ \deff $ of the initial state $ \ket{1}\ket{E^{\rmB}_j} $,
which implies that $ \deff $ increases with $ L $ at all energy scales.
We note that an energy eigenstate $  \ket{E^{\rmB}_j}$ is not  necessarily  typical in the Hilbert space of the energy shell,
and the amount of $ \deff $ is nontrivial.
The numerical result in Fig.~\ref{fig:effective_dimension} suggests that fluctuations of observables around the diagonal ensemble averages are negligible,
and the entropy production will indeed saturate.

\begin{figure}
\includegraphics[width=0.9\linewidth]{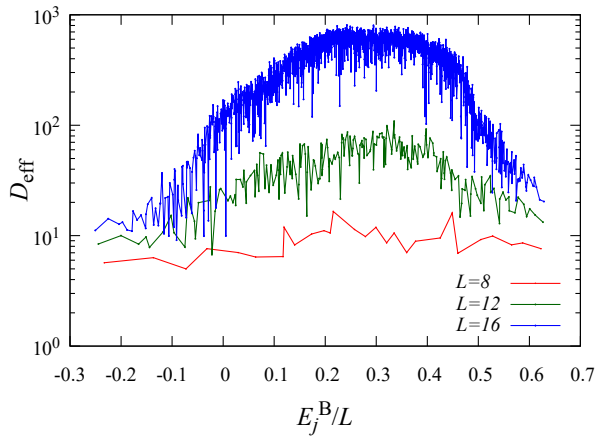}
\caption{\label{fig:effective_dimension}
	(color online)
Initial state dependence of the effective dimensions $ \deff $ for the one-dimensional non-integrable model of hard-core bosons with the bath size $ L=8,12,16 $.
We consider the case that bath $ \mrm{B} $ is $ 1/4 $ filling ($ N_{\mathrm{b}}/L =1/4$) and the initial state is a pure state $ \ket{1}\ket{E^{\rmB}_j} $.
The parameters are given by $ \epsilon/\gamma=1, \gamma_{\mrm{I}}/\gamma=1, g/\gamma=0.1, \gamma'/\gamma=0.2, g'/\gamma=0.1 $.
The effective dimension increases with $ L $.}
\end{figure}

Figure \ref{fig:entropy_production} shows the time dependence of the entropy production.
We set the inverse temperature as $ \beta=0.1 $,
and choose an energy eigenstate $ \ket{E^{\rmB}_j} $
such that its inverse temperature is closest to $ 0.1 $ in the energy shell.
Here, the inverse temperature $ \beta $ of  $ \ket{E^{\rmB}_j} $ is given by the solution of $ \tr_{\mrm{B}}[\hb\ket{E^{\rmB}_j} \bra{E^{\rmB}_j} ] =\tr_{\mrm{B}}[\hb\rbcan(\beta)] $.
We set parameters as $ L=16, \epsilon/\gamma=1,\gamma_{\mrm{I}}/\gamma=1 ,g/\gamma=0.1, \gamma'/\gamma=0.2, g'/\gamma=0.1 $,
and $ \gamma_{\mrm{I}}/\gamma=0.5,1$, or $ 2 $.

As shown in Fig.~\ref{fig:entropy_production},
the entropy production gradually increases up to $ \gamma t\simeq 10^0 $,
oscillates in the medium regime $10^0 \lesssim \gamma t\lesssim 10^2$,
and then saturates in the long time regime $\gamma t\gtrsim 10^2 $.
This behavior is consistent with Theorem \ref{thm:saturation}.
The saturation value in the long time regime is positive for all the cases of $ \gamma_{\mrm{I}}/\gamma=0.5,1,2 $.
This implies that the error term in Theorem \ref{thm:non-negativity} is negligible
in this example.
The error term is also small in a two dimensional hard-core bosons with nearest-neighbor repulsion~\cite{Iyoda2016}.

We note that the saturation value is independent of the interaction rate,
which implies that the interaction strength affects the entropy production only in the short time regime.
The entropy production is also non-negative in the short time regime,
which is out of the scope of the present theory,
but can be explained by our previous theorem 
based on the Lieb-Robinson bound~\cite{Iyoda2016}.

\begin{figure}
\includegraphics[width=0.9\linewidth]{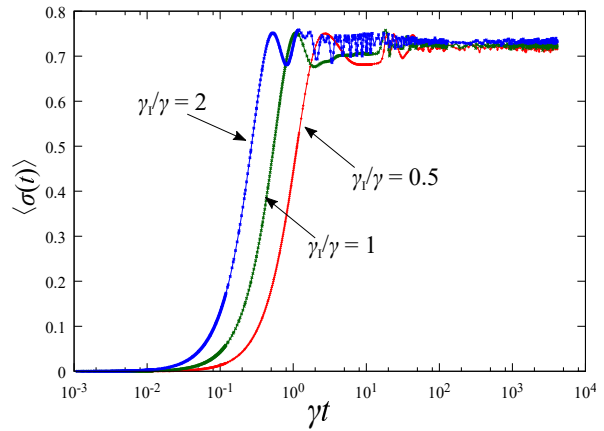}
\caption{\label{fig:entropy_production}
	(color online)
Time dependence of the entropy production with $ L=16 $ and $ N_{\mathrm{b}}=5 $.
The entropy production is plotted against dimensionless time $ \gamma t $,
for three interaction strength between system $ \mrm{S} $ and bath $ \mrm{B} $: $   \gamma_{\mrm{I}}/\gamma=0.5,1,2$.
We set other parameters as $ \epsilon/\gamma=1, g/\gamma=0.1, \gamma'/\gamma=0.2, g'/\gamma=0.1 $.
The initial sate is the product state $ \ket{1}\ket{E^{\rmB}_j} $ and the inverse temperature of the initial energy eigenstate is $ \beta=0.1 $.
The entropy production is non-negative in all time scale,
and saturates in the long time $ \gamma t\gtrsim10^2 $.
We note that the saturation values are almost the same for all the interaction strength $ \gamma_{\mrm{I}} $.}
\end{figure}

\section{\label{sec:Conclusion}Conclusion}
In this paper, we have discussed the second law of thermodynamics for general quantum states in the long time regime.
We have proved Theorem~\ref{thm:saturation},
which states that the entropy production saturates for most times
[inequality~\eqref{eq:saturation}].
The saturation value \eqref{eq:saturation_value} is given by the entropy production with respect to the diagonal ensemble $ \omega $.
We have then derived the lower bound \eqref{eq:lower_bound} of the saturated entropy production.
If the diagonal ensemble is independent of the parameters of the initial state except for the average energy,
the obtained bound \eqref{eq:lower_bound} implies the non-negativity of the saturated entropy production.
This scenario can be justified by assuming the ETH.
We have shown the initial state independence (Lemma~\ref{lem:thermalization})
and the non-negativity of the saturation value (Theorem~\ref{thm:non-negativity}) under the strong ETH (Assumption~\ref{dfn:Strong-ETH}) and an additional assumption on the initial energy (Assumption~\ref{asmp:energy_distribution}).
We have also shown the non-negativity from the weak ETH (Theorem~\ref{thm:non-negativity_weak}).

Our results have clarified that thermodynamic irreversibility can emerge from unitary dynamics of quantum many-body systems in the long time regime.
We emphasize that, if the initial state of the bath is a pure state,
the fluctuation theorem holds only in the short time regime,
but breaks down in the long time regime~\cite{Iyoda2016}.
This is because non-thermal quantum fluctuations of the initial state affect higher-order fluctuations of the entropy production in the long time regime.
In contrast,
as shown in the present work,
the second law is robust
in the long time regime,
because the second law only concerns the first cumulant
(i.e., the average) of the entropy production.
More detailed investigation of the time scale of saturation of the entropy production is a future issue.

\begin{acknowledgments}
We are grateful to Jordan M. Horowitz for his careful reading of the manuscript.
The authors also thank Takashi Mori for valuable discussions.
K.K.  are supported by JSPS KAKENHI Grant Number JP17J06875.
E.I. and T.S. are supported by JSPS KAKENHI Grant Number JP16H02211.
E.I. is also supported by JSPS KAKENHI Grant Number JP15K20944.
T.S. is also supported by JSPS KAKENHI Grant Number JP25103003.
\end{acknowledgments}

\appendix

\section{\label{sec:effective_dimension}Lower bounds of the effective dimension}

We show lower bounds of effective dimensions $ \deff $
for the three cases that
the initial state of bath $ \rmB $ is the microcanonical ensemble $ \rbmc $,
the canonical ensemble $ \rbcan $, 
and a typical pure state in the Hilbert space of the energy shell.
In these cases,
the effective dimensions are exponentially large with respect to the size of bath $ \rmB $.

\subsection{Microcanonical ensemble}
The purity of the microcanonical ensemble $ \rbmc $ is given by
\begin{align}
\tr[(\rbmc)^2] 
&=\frac{1}{\dbmc}
=e^{-N_{\rmB} s(u)},
\end{align}
where $ \dbmc $ is the dimension of the energy shell,
$ N_\rmB $ is the size of bath $ \rmB $,
and $ s(u) $ is the entropy density as a function of energy density $ u $.
From inequality \eqref{eq:bound_eff}, we have the lower bound
\begin{align}
\deff\geq \frac{e^{N_{\rmB} s(u)}}{\nde}.
\label{eq:bound_mc}
\end{align}
Since $ s(u) $ approaches an $ N_\rmB $-independent function in the large bath limit~\cite{Ruelle},
we find that the right-hand side of inequality \eqref{eq:bound_mc} becomes exponentially large in $ N_\rmB $.

\subsection{Canonical ensemble}
The purity of the canonical ensemble $ \rbcan $ is given by
\begin{align}
\tr[(\rbcan)^2] 
&=\frac{\tr[e^{-2\beta\hb}]}{(\tr[e^{-\beta\hb}])^2}
\\
&=\exp\left[-2\beta N_\rmB(f(2\beta)-f(\beta))\right],
\end{align}
where we defined the free energy density by $  f(\beta):=-(\log\tr[e^{-\beta\hb}] ) /\beta N_\rmB$.
From inequality \eqref{eq:bound_eff}, we have the lower bound
\begin{align}
\deff\geq \frac{\exp\left[2\beta N_\rmB(f(2\beta)-f(\beta))\right]}{\nde}.
\label{eq:bound_can}
\end{align}
Since $ f(\beta) $ approaches an $ N_\rmB $-independent function in the large bath limit,
we find that the right-hand side of inequality \eqref{eq:bound_can} becomes exponentially large in $ N_\rmB $.

\subsection{Typical pure state}
We randomly sample a typical pure state $ \ket{\psi} $ from the Hilbert space of the energy shell of bath $ \rmB $ with respect to the Haar measure.
From Theorem 2 of Ref.~\cite{Linden2009},
the probability that the effective dimension is smaller than $ \dbmc/4 $ is bounded from above:
\begin{align}
\prob_{\psi}
\left[\deff<
\frac{\dbmc}{4}\right]
\leq
2\exp(-c\sqrt{\dbmc}).
\end{align}
This inequality implies that the effective dimension of a typical pure state
becomes exponentially large with increasing $ N_\rmB $.

%

\end{document}